\documentclass[a4paper,onecolumn,11pt, accepted=2020-05-08]{quantumarticle}
\pdfoutput =1

\usepackage[T1]{fontenc}
\usepackage[utf8]{inputenc}
\usepackage{palatino,amsmath,amsfonts,amsthm,amssymb}
\usepackage{mathtools}
\usepackage{subeqnarray}
\usepackage{setspace}
\usepackage{graphics,graphicx,color}
\usepackage{url}
\usepackage{enumerate}
\usepackage{float} 
\usepackage{multirow}
\usepackage{hhline}
\usepackage[margin=1.2in]{geometry}
\usepackage{braket}
\usepackage{dsfont} 
\usepackage{authblk}
\usepackage{multirow}
\usepackage{hhline}
\usepackage[makeroom]{cancel}
\usepackage{authblk}
\usepackage{diagbox}
\usepackage[math]{cellspace}
\usepackage{hyperref}
\RequirePackage{doi}
            
\newcommand{\C}{\mathbb C}

\newtheorem{theorem}{Theorem}
\newtheorem{cor}{Corollary}
\newtheorem{definition}{Definition}

\newtheorem{prop}{Proposition}

\newtheorem{example}{Example}

\usepackage[scr=boondoxo,scrscaled=1.05]{mathalfa}

\newcommand{\beq}{\begin{eqnarray}}
\newcommand{\eeq}{\end{eqnarray}}

\newcommand{\proj}[1]{ \ket{#1}\!\bra{#1} }

\newcommand{\saz}{\sigma_\alpha^z}
\newcommand{\sax}{\sigma_\alpha^x}
\newcommand{\m}[1]{\mathcal{#1}}

\begin{document}
\title{A two-player dimension witness based on embezzlement, and an elementary proof of the non-closure of the set of quantum correlations}

\author{Andrea Coladangelo}

\affil{Computing and Mathematical Sciences,	
Caltech\\ 	
{acoladan@caltech.edu}
}
\date{}

\sloppy

\maketitle

\begin{abstract}We describe a two-player non-local game, with a fixed small number of questions and answers, such that an $\epsilon$-close to optimal strategy requires an entangled state of dimension $2^{\Omega(\epsilon^{-1/8})}$. Our non-local game is inspired by the three-player non-local game of Ji, Leung and Vidick \cite{ji2018three}. It reduces the number of players from three to two, as well as the question and answer set sizes. Moreover, it provides an (arguably) elementary proof of the non-closure of the set of quantum correlations, based on embezzlement and self-testing. In contrast, previous proofs \cite{slofstra2019set, dykema2017non, musat2018non} involved representation theoretic machinery for finitely-presented groups and $C^*$-algebras. 
\end{abstract}


\section{Introduction}
A non-local game consists of a one-round interaction between a trusted referee, or verifier, (who asks questions) and two or more spatially isolated players, or provers (who provide answers). The players are cooperating to win the game, but spatial isolation prevents them from communicating. Bell's theorem \cite{Bell:64a}, a landmark result in physics, asserts that there exist games for which players who share entanglement can outperform players who do not, the most famous example being the CHSH game \cite{clauser1969proposed}. The most immediate application of non-local games is to ``test quantumness'': a referee who observes a winning probability in a non-local game which exceeds what is attainable classically can have high confidence that the players (or devices) she is interacting with were sharing entanglement. A more refined analysis of non-local games allows the referee to obtain more precise characterizations of the devices involved. For example, in some cases, the referee might be able to infer that the devices must share \textit{high-dimensional} entanglement \cite{brunner2008testing}. In special cases, the referee might even be able to completely characterize the quantum state inside the devices and the measurements that they are performing (up to local isometries and some small error) \cite{mayers2004self, mckague2012robust, coladangelo2017all}.

In this work, we focus on the study of non-local games as witnesses of high-dimensional entanglement. This has had on the one hand fruitful applications in quantum cryptography, and on the other it has shed light on foundational questions in the theory of entanglement. Before proceeding further, we clarify that when we use the term $\textit{non-local game}$ we do not restrict ourselves to games with a binary outcome (``win'' or ''lose''), but rather we consider games specified by an arbitrary function $V$ taking values in $\mathbb{R}$ which determines the players' score as a function of questions and answers. There is a natural one-to-one correspondence between such non-local games and Bell inequalities. 

\paragraph{Certifying high-dimensional entanglement - previous work and state of the art} Non-local games with the property that a lower bound on the score provides a lower bound on the dimension of the players' quantum systems are referred to as \textit{dimension witnesses}. The study of games (or correlations) with such a property was initiated by Brunner et al. \cite{brunner2008testing}, who coined the term. In this work, we focus on dimension witnesses that can certify entanglement of arbitrarily high dimension.

The first example of a game which cannot be won perfectly with any finite amount of entanglement was proposed by Leung, Toner and Watrous \cite{leung2013coherent}, and is intimately connected to our result. The game that they introduced is not a non-local game in the usual sense, since it involves \textit{quantum} questions and answers. However, it has the property that in order to succeed with high probability, the players have to perform a coherent state exchange which requires them to share an embezzling state of high dimension. More precisely, the game forces the two players to coherently transform a product state of two qubits into an EPR pair, using only local operations. This task is, of course, impossible to perform exactly, but can be performed to arbitrarily high precision if the two players share an auxiliary entangled state of sufficiently high dimension (referred to as an \textit{embezzling} state). 

Subsequently, several examples of dimension witnesses for entanglement of arbitrarily high dimension have been proposed over the years consisting of non-local games with classical questions and answers \cite{briet2011generalized, slofstra2011lower, brunner2013dimension, mancinska2014unbounded, chao2016test, coudron2016parallel, coladangelo2017parallel, coladangelo2017robust, natarajan2017quantum, coladangelo2017separation}. However, all of these examples involve \textit{families} of non-local games whose questions and answers increase as the witnessed dimension increases. For some time, it was an open question to determine whether there exists a non-local game, with a finite number of questions and answers, whose optimal value cannot be attained by any finite-dimensional strategy (in the tensor product model), but which can be attained in the limit of finite-dimensional strategies. This question was answered recently by Slofstra in a sequence of two breakthrough works \cite{slofstra2020tsirelson, slofstra2019set}, where he introduces novel techniques based on the representation theory of finitely-presented groups. Slofstra's result implies that the set of quantum correlations is not closed.

An alternative proof of the latter result was given subsequently by Dykema, Paulsen and Prakash \cite{dykema2017non}, and more recently by Musat and R\o rdam \cite{musat2018non}, using techniques based on the representation theory of $C^*$-algebras. The games constructed in \cite{dykema2017non} and $\cite{musat2018non}$ have significantly smaller question and answer set sizes, namely $5$ and $2$.

In contrast, a more recent result by Coladangelo and Stark \cite{coladangelo2018unconditional} gives an example of a point in the set of quantum correlations on question sets of size $5$ and answer sets of size $3$ which cannot be attained using finite-dimensional entanglement but $\textit{can}$ be attained exactly using infinite-dimensional entanglement, in the tensor product model. This asserts that the  the set $\mathcal{C}_q$ of quantum correlations attainable with finite-dimensional entanglement is strictly contained in the set $\mathcal{C}_{qs}$ of correlations attainable with possibly infinite-dimensional entanglement.

All of the above results are not explicit or quantitative about the tradeoff between winning probability (or expected score in the game) and the dimension required to attain it. What we desire from a dimension witness is a quantitative statement of the following form: if the players' score is $\epsilon$-close to optimal, then their strategy has dimension at least $f(\epsilon)$, where $f(\epsilon)$ is a function that tends to infinity as $\epsilon$ tends to zero. In \cite{slofstra2018entanglement}, Slofstra and Vidick analyze such a tradeoff for the machinery introduced by Slofstra in \cite{slofstra2020tsirelson}, and they relate such tradeoff to a quantity called the \textit{hyperlinear profile} of a group. In a subsequent work \cite{slofstra2018group}, Slofstra provides a finitely-presented group whose hyperlinear profile is at least subexponential. As a corollary, this yields a two-player non-local game, with question and answer sets of finite size, with the property that a $1-\epsilon$ winning probability requires dimension at least $2^{\Omega(\epsilon^{-c})}$ to attain for some constant $0<c<1$. The caveat of such a non-local game is that its description is quite involved and the size of question and answer sets is large. Moreover, it is not clear whether a winning probability of $1$ in the game can be attained in the limit of finite-dimensional strategies or not (although it can be attained in the commuting-operator model). These caveats not only make an experimental demonstration of such a dimension witness infeasible, but, more importantly, they somewhat conceal what is truly happening behind the scenes: the resulting non-local game, although remarkable for its behaviour, does not arguably provide much intuition about what is causing the exponential blow-up of the dimension.

A much simpler game with a similar exponential tradeoff between optimality and dimension, and without this caveat, but involving three players, was proposed recently by Ji, Leung and Vidick $\cite{ji2018three}$. Their work constitutes, in some sense, a return to the original ideas of Leung, Toner and Watrous'coherent state-exchange game \cite{leung2013coherent}, which are cleverly translated to a setting in which all questions and answers are classical. At the heart of the three-player non-local game of Ji, Leung and Vidick is the idea of delegating the actions of the \textit{quantum} verifier of the coherent state-exchange game to a third player. By combining different non-local tests, the verifier is still able, using only classical communication, to enforce that two of the three players must be performing a coherent state-exchange which involves a high-dimensional embezzling state as a resource. 

\paragraph{Our result} In this work, we show, strikingly, that the third player is not required. We design a much more direct two-player non-local game with an (improved) exponential trade-off between optimality and dimension: one of the key ideas is the introduction of a simple additional sub-test which can guarantee the coherence of a state-exchange between the two players even in the absence of a ``physical'' third register that forces coherence, like in the games of \cite{leung2013coherent} and \cite{ji2018three}. Our result is the following:

\begin{theorem}{(informal)}
There exists a two-player non-local game on question sets of size $5$ and $6$, and answer sets of size $3$, with the property that:

\begin{itemize}
    \item (completeness) For any $\epsilon > 0$, there exists a strategy of dimension $2^{O(\epsilon^{-1})}$ that is $\epsilon$-close to optimal.
    \item (soundness) Any $\epsilon$-close to optimal strategy has at least $2^{\Omega(\epsilon^{-1/8})}$ dimension.
\end{itemize}
\end{theorem}

Our game can be thought of as a direct \textit{de-quantization} of the coherent state-exchange game. It is by far the simplest non-local game (in terms of question and answer set size) with such an exponential tradeoff. For a comparison, even with three players, the question and answer sets are of size $12$ and $8$ respectively in \cite{ji2018three}.  

Our game provides a new proof of the non-closure of the set of quantum correlations. However, strikingly, compared to the proofs in \cite{slofstra2019set}, \cite{dykema2017non} and \cite{musat2018non}, our proof is arguably elementary, and does not involve any representation-theoretic machinery. We point out, additionally, that an exponential tradeoff between optimality and dimension does not hold for the game in \cite{dykema2017non}, where a strategy of dimension $1/poly(\epsilon)$ can be $\epsilon$-close to optimal (and we suspect that this is also the case for the game in \cite{musat2018non}). 

Next, we sketch the main ideas in the design of our two-player non-local game.

\subsection{A sketch of our two-player non-local game}
For ease of exposition, we will leave states unnormalized in this introduction.

Our game consists of sub-tests (a), (b) and (c), executed by the verifier with equal probability:

\begin{itemize}
    \item[(a)] A non-local game $G_{\text{3-CHSH}}$ whose unique optimal strategy requires the provers to share the state $\ket{00} + \ket{11} + \ket{22}$. $G_{\text{3-CHSH}}$ is an instance (for $d=3$) of a more general family of non-local games from \cite{col2018generalization}. $G_{\text{3-CHSH}}$ contains a special ``computational basis'' question for Alice which requires her to measure her half of the state in the computational basis.
    \item[(b)] The well-known ``tilted CHSH'' non-local game, which we denote by $G_{\text{tCHSH}}$ \cite{Acin12, bamps2015sum}. This requires, for the appropriate choice of parameters, that the provers share the state $\ket{00}+ \sqrt{2}\ket{11}$. $G_{\text{tCHSH}}$ contains a special ``computational basis'' question for Bob, which requires him to measure in the computational basis.
    \item[(c)] A sub-test in which Alice is asked the ``computational basis'' question from (a), and Bob is asked the ``computational basis'' question from (b). Alice and Bob win if: either they both answer ``0'', or they both answer different from ``0''.
\end{itemize}

The intuition behind the game is the following: Alice and Bob could share the state  $(\ket{00} + \ket{11} + \ket{22})_{\textsf{AB}} \otimes (\ket{00}+ \sqrt{2}\ket{11})_{\textsf{A'B'}}$. This would allow them to win parts (a) and (b) optimally, but they would fail in part (c). The power of part (c) is that Alice is uncertain about whether she is being asked a question from part (a) or (c), and Bob is uncertain about whether he is being asked a question from part (b) or (c). Magically, the condition of part (c) is sufficient to enforce that Alice and Bob cannot keep the two optimal states from part (a) and (b) into two separate registers, but rather they should coherently transform one into the other in order to achieve consistency in answering part (c). This coherent transformation is what requires an exponentially growing amount of entanglement dimension to perform to increasing precision. We refer the reader to section \ref{sec: our non-local game} for a formal description of our game.

\paragraph{Outline} Section \ref{sec: prelim} covers some preliminaries: \ref{sec: notation} introduces some notation, \ref{sec: non-local games} introduces non-local games, \ref{sec: non-local games examples} gives two important examples of non-local games which are used as sub-tests in our non-local game, and \ref{sec: embezzlement} briefly introduces embezzlement. Section \ref{sec: our non-local game} describes our non-local game. Section \ref{sec: completeness} covers completeness: we give a family of strategies that approximates arbitrarily well the optimal value in our non-local game. Section \ref{sec: soundness} covers soundness: we show that any close to optimal strategy requires a lot of entanglement. Section \ref{sec: non-closure} briefly discusses how our non-local game implies the non-closure of the set of quantum correlations.

\section{Preliminaries}
\label{sec: prelim}
\subsection{Notation}
\label{sec: notation}
For some unitary $U$, $\delta >0$, states $\ket{\psi}, \ket{\phi}$, we write $\ket{\psi} \approx_{U, \delta} \ket{\phi}$ if $\|U\ket{\psi} - \ket{\phi} \| \leq \delta$, where $\|\cdot\|$ is the Euclidean norm. We write $\ket{\psi} \approx_{\delta} \ket{\phi}$ if $\|\ket{\psi} - \ket{\phi} \| \leq \delta$.
Let $\sigma^x, \sigma^y, \sigma^z$ be the Pauli matrices.

\subsection{Non-local games}
\label{sec: non-local games}
\begin{definition}[Non-local game]
A non-local game $G$ is a tuple $G = (\mathcal{X}, \mathcal{Y}, \mathcal{A}, \mathcal{B}, D, V)$, where $\mathcal{X}, \mathcal{Y}, \mathcal{A}, \mathcal{B}$ are sets, $D$ is a distribution over $\mathcal{X}\times \mathcal{Y}$, and $V: \mathcal{X} \times \mathcal{Y} \times \mathcal{A} \times \mathcal{B} \rightarrow \mathbb{R}$.  $\mathcal{X}$ and $\mathcal{Y}$ are referred to as question sets, and $\mathcal{A}$ and $\mathcal{B}$ as answer sets. $V$ is referred to as the scoring function.
\end{definition}
We denote by $D(x,y)$ the probability of outcome $x,y$ according to distribution $D$.
Note that we use the term \textit{non-local game} to refer to games in which the scoring function $V$ can take any real value, not just values in $\{0,1\}$ like is sometimes the case in the literature. With this nomenclature, non-local games and Bell inequalities are equivalent.

\begin{definition}[Quantum strategy for a non-local game] A quantum strategy for a non-local game $G = (\mathcal{X}, \mathcal{Y}, \mathcal{A}, \mathcal{B}, D, V)$ is a triple $$\left(\ket{\Psi} \in \mathcal{H}_A\otimes \mathcal{H}_B, \{P_x^a: a\in \mathcal{A}\}_{x \in \mathcal{X}}, \{Q_y^b: b \in \mathcal{B}\}_{y \in \mathcal{Y}}\right),$$ where $\mathcal{H}_{A} , \mathcal{H}_B$ are Hilbert spaces, $\{P_x^a: a\in \mathcal{A}\}_{x \in \mathcal{X}}$ is a set of projective measurements on $\mathcal{H}_{A}$, and $\{Q_y^b: b \in \mathcal{B}\}_{y \in \mathcal{Y}}$ on $\mathcal{H}_B$.
\end{definition}

\begin{definition}[Value of a quantum strategy in a game]
Let $G = (\mathcal{X}, \mathcal{Y}, \mathcal{A}, \mathcal{B}, D, V)$ be a non-local game, and $S = (\ket{\Psi} \in \mathcal{H}_A\otimes \mathcal{H}_B, \{P_x^a: a\in \mathcal{A}\}_{x \in \mathcal{X}}, \{Q_y^b: b \in \mathcal{B}\}_{y \in \mathcal{Y}})$ a quantum strategy for $G$. The value of $S$ in $G$ is $$\omega(S,G) := \sum_{x \in \m X, y \in \m Y} D(x,y) \cdot V(x,y,a,b) \cdot \bra{\Psi}P_x^a \otimes Q_{y}^b \ket{\Psi}$$
\end{definition}
Note that the value $\omega(S,G)$ corresponds to the expected score of strategy $S$ in game $G$, assuming that questions are distributed according to $D$, and that the score is determined by the function $V$.

\begin{definition}[Quantum value of a game]
The quantum value $\omega^*(G)$ of a game $G = (\mathcal{X}, \mathcal{Y}, \mathcal{A}, \mathcal{B}, D, V)$ is defined as follows:
$$\omega^*(G) := \sup_S{\omega(S,G)},$$
where the supremum is taken over all quantum strategies for $G$.
\end{definition}
Since the closure of the set of finite-dimensional quantum correlations contains the set of infinite-dimensional quantum correlations \cite{scholz2008tsirelson}, it does not matter whether the supremum in the definition of $\omega^*$ is taken over finite or infinite-dimensional strategies (i.e. whether $\m H_A$ and $\m H_B$ are finite or infinite-dimensional).

Finally, we introduce some terminology which we will primarily employ in section \ref{sec: non-closure}.
\begin{definition}[Correlation]
Given sets $\m{X}$,$\m{Y}$,$\m{A}$,$\m{B}$, a (bipartite) \textit{correlation} is a collection $\{p(a,b|x,y): a\in \m{A}, b \in \m{B}\}_{(x,y)\in \m{X}\times \m{Y}}$, where each $p(\cdot,\cdot|x,y)$ is a probability distribution over $\m A\times \m B$.
\end{definition}

\begin{definition}[Correlation induced by a quantum strategy]
Given a quantum strategy $S = (\ket{\Psi} \in \mathcal{H}_A\otimes \mathcal{H}_B, \{P_x^a: a\in \mathcal{A}\}_{x \in \mathcal{X}}, \{Q_y^b: b \in \mathcal{B}\}_{y \in \mathcal{Y}})$ for a non-local game $G = (\mathcal{X}, \mathcal{Y}, \mathcal{A}, \mathcal{B}, D, V)$, the correlation induced by $S$ is $\{p(a,b|x,y): a\in \m{A}, b \in \m{B}\}_{(x,y)\in \m{X}\times \m{Y}}$ where, for all $a,b,x,y,$ $$p(a,b|x,y) = \bra{\Psi}P_x^a \otimes Q_{y}^b \ket{\Psi}.$$
\end{definition}

\begin{definition}[Value of a correlation in a game]
Let $G = (\mathcal{X}, \mathcal{Y}, \mathcal{A}, \mathcal{B}, D, V)$ be a non-local game, and $p =\{p(a,b|x,y): a\in \m{A}, b \in \m{B}\}_{(x,y)\in \m{X}\times \m{Y}}$ a correlation. The value $\omega(p, G)$ of $p$ in $G$ is defined as 
$$ \omega(p, G) := \sum_{x \in \m X, y \in \m Y} D(x,y) \cdot V(x,y,a,b) \cdot p(a,b|x,y).$$
\end{definition}

Clearly, if $p$ is the correlation induced by a quantum strategy $S$ for game $G$, then $\omega(p,G) = \omega(S,G)$.

We note that quantum strategies can also be considered separately from non-local games (i.e. without specifying any distribution $D$ or scoring function $V$). We then denote by $\mathcal{C}_q$ (resp. $\m C_{qs}$) the set of correlations induced by a finite-dimensional (resp. possibly infinite-dimensional) quantum strategy. We denote by $\m C_{qa}$ the closure of $\m C_q$, which by \cite{scholz2008tsirelson} is known to also be the closure of $\m C_{qs}$.

\subsection{Useful examples of non-local games}
In this section, we describe two families of non-local games which we will employ as sub-tests in our non-local game.
\label{sec: non-local games examples}
\paragraph{Tilted CHSH}

We introduce the tilted CHSH inequality \cite{Acin12}, which is a building block for the non-local game in this work. First, we recall the CHSH inequality. It states that for binary observables $A_0, A_1$ on Hilbert space $\mathcal{H}_A$ and binary observables $B_0, B_1$ on Hilbert space $\mathcal{H}_B$ together with a product state $\ket\phi = \ket{\phi_A} \otimes \ket{\phi_B}$, we have
\begin{equation}
	\bra{\phi} A_0B_0 + A_0B_1 +A_1B_0 - A_1B_1 \ket{\phi} \leq 2,
\end{equation}
where the maximum is achieved (for example setting all observables to identity). However, if instead of the product state $\ket\phi$ we allow an entangled state $\ket\psi$, then the right-hand side of the inequality increases to $2\sqrt2$. This maximum requires a maximally entangled pair of qubits to achieve. In this work, we would like to use an inequality that requires a non-maximally entangled state to achieve the maximum; this is the tilted CHSH inequality. Given a real parameter 
$\beta \in [0,2)$, for a product state $\ket{\phi} = \ket{\phi_A} \otimes \ket{\phi_B}$,
\begin{equation}
\bra{\phi} \beta A_0 + A_0B_0 + A_0B_1 +A_1B_0 - A_1B_1 \ket{\phi} \leq 2+\beta.
\end{equation}
For entangled $\ket\psi$, we have instead that
\begin{equation}\label{eq:tiltedchsh}
	\bra{\psi} \beta A_0 + A_0B_0 + A_0B_1 +A_1B_0 - A_1B_1 \ket{\psi}\leq \sqrt{8+2\beta^2}.
\end{equation}
The maximum in the tilted CHSH inequality is attained by the following strategy:
\begin{definition}[Ideal strategy for tilted CHSH]
\label{def: ideal tilted chsh}
Given parameter $\beta \in [0,2)$, let $\theta \in (0, \frac{\pi}{4}]$ be such that $\sin 2\theta = \sqrt{\frac{4-\beta^2}{4+\beta^2}}$, $\mu = \arctan \sin 2\theta$, and $\alpha = \tan \theta$. Define the \emph{$\alpha$-tilted Pauli operators} as
\begin{equation}
 	\saz := \cos \mu \sigma^z + \sin \mu \sigma^x\text{, and } \sax := \cos \mu \sigma^z - \sin \mu \sigma^x.
 \end{equation} 
The ideal strategy for tilted CHSH with parameter $\beta$ (i.e.\ achieving maximal violation of \eqref{eq:tiltedchsh}) consists of the joint state $\ket{\Psi} = \cos \theta (\ket{00} + \alpha \ket{11})$ and observables $A_0, A_1$ and $B_0, B_1$ with $A_0 = \sigma^z$, $A_1 = \sigma^x$, $B_0 = \sigma_\alpha^z$ and $B_1 = \sigma_\alpha^x$.

\end{definition}

$\beta$ and $\alpha$ are related by an invertible function, and $\alpha$ is typically the parameter of interest, so we choose to denote by $\textnormal{tCHSH}(\alpha)$ the tilted CHSH game whose ideal state is $\ket{\Psi} = \cos \theta (\ket{00} + \alpha \ket{11})$.

We can equivalently formulate the tilted CHSH inequality as a non-local game, as follows:
\begin{definition}[Tilted CHSH as a non-local game]
\label{def: tchsh as game}
For $\alpha \in (0,1]$, the tilted CHSH game $G_{\textnormal{tCHSH}(\alpha)}$ is 
\begin{equation}
G_{\textnormal{tCHSH}(\alpha)} = (\mathcal{X}, \mathcal{Y}, \mathcal{A}, \mathcal{B}, D, V_{\textnormal{tCHSH}(\alpha)}),
\end{equation}
where $\m X, \m Y, \m A, \m B = \{0,1\}$, $D$ is uniform on $\m X \times \m Y$, and $V_{\textnormal{tCHSH}(\alpha)} = (-1)^{a \oplus b - xy} + \delta_{\{x=y=0\}} \cdot \beta \cdot (-1)^a$, where $\beta$ and $\alpha$ are related as in Definition \ref{def: ideal tilted chsh}. 
\end{definition}

\begin{prop}[Quantum value of the tilted CHSH game]
For $\alpha \in (0,1]$, the value of $G_{\textnormal{tCHSH}(\alpha)}$ is $\omega^*_{\textnormal{tCHSH}(\alpha)} := \frac14\cdot\sqrt{8+2\beta^2}$, where $\beta$ and $\alpha$ are related as in Definition \ref{def: ideal tilted chsh}. 
\end{prop}
\begin{proof}
Notice that for any strategy $S$, the value $\omega(S, G_\textnormal{tCHSH}(\alpha))$ takes precisely the form of the LHS of \eqref{eq:tiltedchsh} (upon associating, for each observable in \eqref{eq:tiltedchsh}, the projection onto the $+1$-eigenspace with answer $0$ and the projection onto the $-1$-eigenspace with answer $1$, and up to a factor of $\frac14$ from sampling the questions uniformly).
\end{proof}
In other words, the LHS of the tilted CHSH inequality and the value of the tilted CHSH game are equivalent reformulations of one another. The following theorem asserts a robust self-testing result for tilted CHSH, i.e. that any strategy that attains a value close to the quantum value of the game, must be close to the ideal strategy of Definition \ref{def: ideal tilted chsh} (in the following statement we only write down the conditions that we make use of later).

\begin{theorem}[Self-testing with tilted CHSH (\cite{yang2013robust, bamps2015sum})]
\label{thm: tilted chsh self testing}
Let $\alpha \in (0,1]$. Maximal value in $G_{\textnormal{tCHSH}(\alpha)}$ self-tests the ideal strategy of Definition \ref{def: ideal tilted chsh} with robustness $O(\sqrt{\epsilon})$, i.e. for any strategy $S = (\ket{\Psi} \in \m H_A \otimes \m H_B, \{P_x^a\}, \{Q_y^b\}) $ with value $\omega(S, G_{\textnormal{tCHSH}(\alpha)}) >\omega^*_{\textnormal{tCHSH}(\alpha)}- \epsilon$ there exists a local unitary $U$ and an auxiliary state $\ket{aux}$ such that:
    \begin{itemize}
        \item $\ket{\Psi} \approx_{U, O(\epsilon^{1/2})} \frac{1}{\sqrt{1+\alpha^2}} (\ket{00}+\alpha \ket{11}) \otimes \ket{aux}$
        \item $P_{0}^0\ket{\Psi} \approx_{U, O(\epsilon^{1/2})} \frac{1}{\sqrt{1+\alpha^2}}\ket{00} \otimes \ket{aux}$
    \end{itemize}
\end{theorem}
The last condition means that the first player's measurement on question ``$0$'' is equivalent (up to a change of basis) to a computational basis measurement.

For clarity of notation and exposition in later sections, it is convenient for us to define the game $G_{\sim\textnormal{tCHSH}(\alpha)}$, for $\alpha \in (0,1]$. This is an equivalent version of $G_{\textnormal{tCHSH}(\alpha)}$ with the only difference that the scoring function is $V_{\sim\textnormal{tCHSH}(\alpha)} := (-1)^{a \oplus b - xy} - \delta_{\{x=y=0\}} \cdot \beta \cdot (-1)^{a}$ (notice the minus sign). It is easy to see that this game is equivalent to the original tilted CHSH up to a flip of the answer labels (so in particular $\omega^*_{\textnormal{tCHSH}(\alpha)} = \omega^*_{\sim\textnormal{tCHSH}(\alpha)})$. The corresponding version of Theorem \ref{thm: tilted chsh self testing} for $G_{\sim\textnormal{tCHSH}(\alpha)}$ is as follows:
\begin{theorem}
\label{thm: tilted chsh self testing sim}
Let $\alpha \in (0,1]$. Maximal value in $G_{\sim\textnormal{tCHSH}(\alpha)}$ self-tests the ideal strategy of Definition \ref{def: ideal tilted chsh} with robustness $O(\sqrt{\epsilon})$, i.e. for any strategy $S = (\ket{\Psi} \in \m H_A \otimes \m H_B, \{P_x^a\}, \{Q_y^b\}) $ with value $\omega(S, G_{\sim\textnormal{tCHSH}(\alpha)}) >\omega^*_{\sim\textnormal{tCHSH}(\alpha)}- \epsilon$ there exists a local unitary $U$ and an auxiliary state $\ket{aux}$ such that:
    \begin{itemize}
        \item $\ket{\Psi} \approx_{U, O(\epsilon^{1/2})} \frac{1}{\sqrt{1+\alpha^2}} (\alpha\ket{00}+ \ket{11}) \otimes \ket{aux}$
        \item $P_{0}^0\ket{\Psi} \approx_{U, O(\epsilon^{1/2})} \frac{\alpha}{\sqrt{1+\alpha^2}}\ket{00} \otimes \ket{aux}$
    \end{itemize}
\end{theorem}

To provide a little more insight, the reason why we introduce this equivalent version of tilted CHSH is the following. A building block in our non-local game will be a game that self-tests the state $\frac{1}{\sqrt{5}}\left(\ket{00} + \sqrt{2}\ket{11}\right)$, i.e. the ideal state for tilted CHSH with $\alpha = \sqrt{2}$. However, the parameter $\alpha$ in the tilted CHSH game is required to be in $(0,1]$. Of course, the point is that the tilted EPR pair with $\alpha = \sqrt{2}$ is equivalent, up to a $\sigma_x$ on both tensor factors, to the tilted EPR pair with $\alpha = \frac{1}{\sqrt{2}}$. The alternative version of tilted CHSH from Theorem \ref{thm: tilted chsh self testing sim} ensures that the answer labels are naturally consistent with the self-tested state being $\frac{1}{\sqrt{5}}\left(\ket{00} + \sqrt{2}\ket{11}\right)$ rather than $\frac{1}{\sqrt{5}}\left(\ket{11} + \sqrt{2}\ket{00}\right)$.

This is a notational choice that simplifies the description of our non-local game in Section \ref{sec: our non-local game}.


\paragraph{Generalization of CHSH self-testing states of local dimension $d$} There is a family of non-local games, parametrized by $d\geq 2 \in \mathbb{N}$, which generalizes the CHSH game \cite{col2018generalization}. The games in this family have the property that, for the game with parameter $d$, maximal score in the game self-tests the maximally entangled state of local dimension $d$. Each of the games in this family is a 2-player game in which question sets are of size $2 + \mathds{1}_{d>2}$ and $2+ 2\cdot \mathds{1}_{d>2}$, and answer sets are of size $d$. When $d=2$, the game coincides with the usual CHSH game. We denote by $G_{d\textnormal{-CHSH}}$ the game in the family with parameter $d$. We do not describe this family of games in full detail here (for details we refer to \cite{col2018generalization}). We will just recall the self-testing properties of the game that we need in the following theorem, and describe the ideal strategy for the case of $d=3$ (we will use $G_{3\textnormal{-CHSH}}$ later as a sub-test in our non-local game).

\begin{theorem}[\cite{col2018generalization}]
\label{thm: self-testing gen chsh}
There exists a family of non-local games $\{G_{d\textnormal{-CHSH}}\}_{d \geq 2 \in \mathbb{N}}$ with the following properties:
\begin{itemize}
    \item Question sets are:
    \begin{itemize}
        \item $\mathcal{X} = \mathcal{Y} = \{0,1\}$, for $d =2$
        \item $\mathcal{X} = \{0,1,2\}, \mathcal{Y} = \{0,1,2,3\}$, for $d >2$.
        
    \end{itemize} 
    Answer sets are $\mathcal{A} = \mathcal{B} = \{0,1,..,d-1\}$. For all $d$, the distribution over questions is uniform. Denote by $V_{d\textnormal{-CHSH}}$ the scoring function for $G_{d\textnormal{-CHSH}}$.
    
    \item (Self-testing) Let $\omega^*_{d\textnormal{-CHSH}}$ be the value of the game with parameter $d$. There exists a constant $C>0$ such that the following holds. Any strategy $S = (\ket{\Psi}, \{P_x^a\}, \{Q_y^b\}) $ with value $\omega(S, G_{d\textnormal{-CHSH}}) \geq \omega^*_{d\textnormal{-CHSH}}- \epsilon$, for some  $0< \epsilon < \frac{C}{d^3}$, is such that there exists a local unitary $U$ and an auxiliary state $\ket{aux}$ such that: 
    \begin{itemize}
        \item $\ket{\Psi} \approx_{U, O(d^6\epsilon^{1/8})} \frac{1}{\sqrt{d}}\sum_{i=0}^{d-1} \ket{ii} \otimes \ket{aux}$
        \item $P_{0}^i\ket{\Psi} \approx_{U, O(d^6\epsilon^{1/8})} \frac{1}{\sqrt{d}}\ket{ii} \otimes \ket{aux}$.
    \end{itemize}
\end{itemize}
\end{theorem}
Again, the last condition means that the first player's measurement on question ``$0$'' is equivalent (up to a change of basis) to a computational basis measurement.

Next, we describe the ideal strategy for $G_{3\textnormal{-CHSH}}$. First, we fix some notation.

We define an isometry $V: (\C^2)_{A} \to (\C^3)_{\tilde{A}}$ as follows:
\begin{equation}
	V \ket 0 = \ket {1},\,\,\,
    V \ket 1 = \ket {2}
\end{equation}

For an operator $O$ on $\mathbb{C}^{2}$, we write $V(O)$ to refer to the pushforward $VO V^\dagger$ of $O$ along $V$. For example, $V (\sigma^z) = \ket{1}\bra{1} - \ket{2}\bra{2}$. If $O$ has $+1,0,-1$ eigenvalues, we write $O^+$ for the projection onto the $+1$ eigenspace and $O^-$ for the projection onto the $-1$ eigenspace. One can check that with this notation $O = O^+ - O^-$.
We use the notation $\bigoplus A_i$ to denote the direct sum of observables $A_i$. If $\m H_A \approx \C^3$, we still write $\sigma^z_A$ to mean $\sigma^z_A = \proj 0_A - \proj 1_A$. On the other hand, in accordance with the notation above, we write $V(\sigma^z)_A$ to mean  $V(\sigma^z)_A = \proj 1 - \proj 2$. We adopt an analogous notation for all other Paulis and tilted Paulis, and projections onto their eigenspaces. (We will make use of the $\alpha$-tilted Paulis $\saz,\sax$ from Definition \ref{def: ideal tilted chsh}).

\begin{definition}[Ideal strategy for $G_{3\textnormal{-CHSH}}$ \cite{col2018generalization}]
\label{def: ideal 3-chsh}
The ideal strategy for $G_{3\textnormal{-CHSH}}$ is $(\ket{\Psi}, \{P_x^a\}, \{Q_y^b\})$, where $\ket{\Psi} = \frac{1}{\sqrt{3}} (\ket{00} + \ket{11} + \ket{22})$, and the ideal measurements are described in Tables \ref{tab: 3chsh ideal alice} and \ref{tab: 3chsh ideal bob}.
\begin{table}[H]
    \caption{Alice's ideal measurements for $G_{3\textnormal{-CHSH}}$. The entry in cell $x,a$ is the projector $P_{x}^a$.}
    \label{tab: 3chsh ideal alice}
    \centering
    \begin{tabular}{| c || Sc | Sc | Sc | Sc | Sc |}
	\hline
	\diagbox[width=2.5em]{$x$}{$a$} & 0 & 1 & 2 \\ \hhline{|====|}
	$0$  & $\proj 0_{\tilde{A}}$ & $\proj 1_{\tilde{A}}$ & $\proj 2_{\tilde{A}}$ \\ \hline
	$1$ & $(\sigma^x)^+$ & $(\sigma^x)^-$ & $\proj 2$ \\ \hline
	$2$ & $\proj 0$ & $[V(\sigma^x)]^+$ & $[V(\sigma^x)]^-$ \\ \hline
\end{tabular}
\end{table}

\begin{table}[H]
    \centering
    \caption{Bob's ideal measurements for $G_{3\textnormal{-CHSH}}$. The entry in cell $y,b$ is the projector $P_y^b$.}
    \label{tab: 3chsh ideal bob}
    \begin{tabular}{| c || Sc | Sc | Sc | Sc | Sc |}
	\hline
	\diagbox[width=2.5em]{$y$}{$b$} & 0 & 1 & 2 \\ \hhline{|====|}
	$0$  & $(\sigma^z_{\alpha = 1})^+$ & $(\sigma^z_{\alpha = 1})^-$ & $\proj 2$ \\ \hline
	$1$  & $(\sigma^x_{\alpha = 1})^+$ & $(\sigma^x_{\alpha = 1})^-$ & $\proj 2$ \\ \hline
	$2$  & $\proj 0$ & $[V(\sigma^z_{\alpha=1})]^+$ & $[V(\sigma^z_{\alpha=1})]^-$ \\ \hline
	$3$  & $\proj 0$ & $[V(\sigma^x_{\alpha=1})]^+$ & $[V(\sigma^x_{\alpha=1})]^-$ \\ 
	\hline
\end{tabular}
\end{table}

\end{definition}

We emphasize, as it will be important later, that both the ideal strategies for $G_{\textnormal{tCHSH}(\alpha)}$ and $G_{d\textnormal{-CHSH}}$ include a computational basis measurement for the first player on question ``$0$''.

\subsection{Embezzlement}
\label{sec: embezzlement}
The phenomenon of \textit{embezzlement} was first discovered by van Dam and Hayden \cite{van2003universal}. A family of embezzling states can be used to coherently transform a product state into an EPR pair (or viceversa). The fidelity of this transformation increases with the dimension of the embezzling state.

\begin{definition}[Embezzlement]
\label{def: embezzlement}
Let $\{\ket{\Gamma_d}\}_{d \in \mathbb{N}}$ be a collection of states, where $\ket{\Gamma_d} \in (\mathbb{C}^2)^{\otimes d}_{A'} \otimes (\mathbb{C}^2)^{\otimes d}_{B'}$. We say that $\{\ket{\Gamma_d}\}_{d \in \mathbb{N}}$ is an ``embezzling family'' if there exist unitaries $W_{AA'}$ on $\mathbb{C}^2_{A} \otimes (\mathbb{C}^2)^{\otimes d}_{A'}$ and $W_{BB'}$ on $\mathbb{C}^2_{B} \otimes (\mathbb{C}^2)^{\otimes d}_{B'}$ such that $$ \| W_{AA'} \otimes W_{BB'} \ket{\textnormal{EPR}}_{AB}\ket{\Gamma_d}_{A'B'} - \ket{11}_{AB}\ket{\Gamma_d}_{A'B'} \| = O\left(\frac{1}{\sqrt{d}}\right)
.$$
\end{definition}

\begin{example}
\label{ex: emb family}
Let $\ket{\Gamma_d} = \frac{1}{\sqrt{N_d}} \sum_{j=1}^d \ket{11}_{A'B'}^{\otimes j}\ket{\textnormal{EPR}}^{\otimes (d - j)}_{A'B'}$, where $N_d$ is a normalizing constant. Then, the family of states $\{\ket{\Gamma_d}\}$ is an embezzling family. The unitaries $W_{AA'}$ and $W_{BB'}$ are the ``left-shift'' unitaries, which act on  $\mathbb{C}^2_{A} \otimes (\mathbb{C}^2)^{\otimes d}_{A'}$ and $\mathbb{C}^2_{B} \otimes (\mathbb{C}^2)^{\otimes d}_{B'}$ respectively, by shifting by one to the left each of the $d+1$ qubit registers. It is easy to check that the family of states $\{\ket{\Gamma_d}\}_{d \in \mathbb{N}}$ satisfies Definition \ref{def: embezzlement}.
\end{example}

\section{Our non-local game}
\label{sec: our non-local game}
The following is our non-local game. We describe it informally first, and then we give a precise description in Fig. \ref{fig: our non-local game}. We refer to Alice and Bob as the two players in our non-local game.

The non-local game consists of three tests, run with equal probability.
\begin{itemize}
    \item[(a)] In the first test, the verifier sends both players questions from the game $G_{3\textnormal{-CHSH}}$, and they obtain a score according to its scoring function.
    \item [(b)] In the second test, the verifier sends both players questions from the (flipped) tilted CHSH game $G_{\sim\textnormal{tCHSH}}(\frac{1}{\sqrt{2}})$. Importantly, their roles are also switched: Alice is sent the questions of player $2$ in $G_{\sim\textnormal{tCHSH}}(\frac{1}{\sqrt{2}})$, and Bob the questions of player $1$. They obtain a score according to the scoring function of $G_{\sim\textnormal{tCHSH}}(\frac{1}{\sqrt{2}})$.
    \item[(c)] In the third test, Alice receives the ``computational basis'' question (question ``$0$'' of the first player) from the game $G_{3\textnormal{-CHSH}}$, and Bob receives the ``computational basis'' question (question ``$0$'' of the first player) from the game $G_{\sim\textnormal{tCHSH}}(\frac{1}{\sqrt{2}})$. The players' score is $1$ if: Alice answers $0$ if and only if Bob answers $0$. They score $0$ otherwise. 
\end{itemize}

The intuition behind this game is the following. 

If Alice and Bob's strategy attains an $\epsilon$-close to optimal expected score overall (where optimally here means playing perfectly in all three tests), then it must attain a $3\epsilon$-close to optimal expected score in each of the three tests. 
By the self-testing result of Theorem \ref{thm: self-testing gen chsh}, in order to play $3\epsilon$-close to optimally in $(a)$, the players need to be sharing a state close to a maximally entangled state of qutrits, up to a local isometry, and moreover one of Alice's measurements is a ``computational basis'' measurement. By Theorem \ref{thm: tilted chsh self testing sim}, in order to play $3\epsilon$-close to optimally in $(b)$, Alice and Bob must be measuring a state close to a tilted EPR pair with ratio $\frac{1}{\sqrt{2}}$, up to a local isometry. Moreover one of Bob's measurements must be a ``computational basis'' measurement. 
Crucially, Alice cannot distinguish her question in $(c)$ from a ``computational basis'' question in $(a)$, while Bob cannot distinguish his question in $(c)$ from a ``computational basis'' question in $(b)$. In order to play close to optimally in $(c)$, Alice and Bob's computational basis measurements need to satisfy a consistency condition. It is this consistency condition that forces the two players to ``agree'' on a computational basis element $\ket{00} \in \mathbb{C}^3_A \otimes \mathbb{C}^3_B$, and to perform a coherent state exchange such that: \begin{equation}
\label{coherent ex}
    \frac{1}{\sqrt{3}}(\ket{00} + \ket{11} + \ket{22})_{AB} \mapsto \frac{1}{\sqrt{3}}(\ket{00} + \sqrt{2}\ket{11})_{AB},
\end{equation}
with $\ket{00}_{AB} \mapsto \ket{00}_{AB}$ and $\frac{1}{\sqrt{2}}(\ket{11} + \ket{22})_{AB} \mapsto \ket{11}_{AB}$. The LHS of \eqref{coherent ex} is the state that the players need in order to play part (a) perfectly, while the RHS is the state that they need to play part (b) perfectly. Part (c) ensures that players have to ``agree'' on the term $\ket{00}$, and this enforces that they must perform coherently the exchange in \eqref{coherent ex} to high accuracy if they are to perform well in all three parts.


Next, we give a precise description of our non-local game $G_{emb}$. We denote by $V_{emb}$ its scoring function. Recall that $V_{3\textnormal{-CHSH}}$ and $V_{\sim\textnormal{tCHSH}}(\frac{1}{\sqrt{2}})$ are the scoring functions for games $G_{3\textnormal{-CHSH}}$ and $G_{\sim\textnormal{tCHSH}}(\frac{1}{\sqrt{2}})$ respectively.

\begin{figure}[H]
\rule[1ex]{16.5cm}{0.5pt}
Question sets: $\mathcal{X} := \Big(\{\textnormal{``$3$-CHSH''}\} \times \{0,1,2\} \Big) \cup \Big(\{\textnormal{``$\sim$tCHSH($\frac{1}{\sqrt{2}}$)''}\} \times \{0,1\} \Big)$, and $\mathcal{Y} := \Big(\{\textnormal{``$3$-CHSH''}\} \times \{0,1,2,3\} \Big) \cup \Big(\{\textnormal{``$\sim$tCHSH($\frac{1}{\sqrt{2}}$)''}\} \times \{0,1\} \Big)$.
Answer sets: $\mathcal{A} = \mathcal{B} = \{0,1,2\}$.\\
The game consists of the following three parts, executed with equal probability.
\begin{itemize}
    \item[(a)] Pick uniformly random $x' \in \{0,1,2\}$ and $y' \in \{0,1,2,3\}$. Send $x = (\textnormal{``$3$-CHSH''}, x')$ to Alice and $y = (\textnormal{``$3$-CHSH''}, y')$ to Bob. Let $a$ and $b$ be the players' answers. The players' score is $V_{emb}(a,b,x, y) = V_{3\textnormal{-CHSH}}(a,b,x',y')$.
    \item[(b)] Pick uniformly random $x' \in \{0,1\}$ and $y' \in \{0,1\}$. Send $x = (\textnormal{``$\sim$tCHSH($\frac{1}{\sqrt{2}}$)''}, x')$ to Alice and $y = (\textnormal{``$\sim$tCHSH($\frac{1}{\sqrt{2}}$)''}, y')$ to Bob. Let $a$ and $b$ be the players' answers. The players' score is $V_{emb}(a,b,x, y) = V_{\sim\textnormal{tCHSH}(\frac{1}{\sqrt{2}})}(b,a,y',x')$ (notice that the roles of the two players is switched in the last expression).
    \item[(c)] Send question $x = (\textnormal{``$3$-CHSH''}, 0)$ to Alice, and question $y = (\textnormal{``$\sim$tCHSH($\frac{1}{\sqrt{2}}$)''}, 0)$ to Bob. Let $a$ and $b$ be the players' answers. Their score is $$V_{emb}(a,b,x,y) = \begin{cases}
    1,    & \text{if } (a,b) \in \{(0,0)\} \cup \left(\{1,2\} \times \{1,2\}\right)\\
    0,  & \text{otherwise}
\end{cases}$$
\end{itemize}

\rule[2ex]{16.5cm}{0.5pt}\vspace{-.5cm}
\caption{Our non-local game $G_{emb}$}
  \label{fig: our non-local game}
\end{figure}

\begin{prop}
\label{prop: value}
The value of the non-local game $G_{emb}$ of Fig. \ref{fig: our non-local game} is $\omega^*(G_{emb}) = \frac13 (\omega^*_{3\textnormal{-CHSH}}+\omega^*_{\sim\textnormal{tCHSH}(\frac{1}{\sqrt{2}})} + 1)$.
\end{prop}
\begin{proof}
Clearly, $\omega^* \leq \frac13 (\omega^*_{3\textnormal{-CHSH}}+\omega^*_{\sim\textnormal{tCHSH}(\frac{1}{\sqrt{2}})} + 1)$. Otherwise, there would exist a strategy $S$ such that the value $\omega(S,G_{emb})> \frac13 (\omega^*_{3\textnormal{-CHSH}}+\omega^*_{\sim\textnormal{tCHSH}(\frac{1}{\sqrt{2}})} + 1)$. This would imply that at least one of the following holds:
\begin{itemize}
    \item The restriction of $S$ to part (a) has value greater than $\omega^*_{3\textnormal{-CHSH}}$.
    \item The restriction of $S$ to part (b) has value greater than $\omega^*_{\sim\textnormal{tCHSH}}$.
    \item The restriction of $S$ to part (c) has value greater than $1$.
\end{itemize}
All three of the above are clearly impossible. 

On the other hand, we will construct in the next section a sequence of strategies whose value in $G$ gets arbitrarily close to $\frac13 (\omega^*_{3\textnormal{-CHSH}}+\omega^*_{\sim\textnormal{tCHSH}(\frac{1}{\sqrt{2}})} + 1)$. This completes the proof.
\end{proof}

\section{Completeness}
\label{sec: completeness}
In this section, we describe a family of strategies whose value in our non-local game $G_{emb}$ gets arbitrarily close to $\frac13 (\omega^*_{3\textnormal{-CHSH}}+\omega^*_{\sim\textnormal{tCHSH}(\frac{1}{\sqrt{2}})} + 1)$ (which also completes the proof of Proposition \ref{prop: value}).
A strategy in the family is parametrized by $d \in \mathbb{N}$. The provers start with the state \begin{equation}
\label{eq: starting state}
    \frac{1}{\sqrt{3}}(\ket{00}+\ket{11}+\ket{22})_{\tilde{A}\tilde{B}} \otimes \ket{\Gamma_d}_{A'B'}
\end{equation}
where $\ket{\Gamma_d}$ is an embezzling state. 
We give first an informal description of the ideal measurements, and we follow this by a formal description. 
\begin{itemize}
    \item Upon receiving a question with prefix \textnormal{``3-CHSH''}, Alice and Bob perform the corresponding ideal measurement for \textnormal{3-CHSH}. In particular on question $(\textnormal{``3-CHSH''}, 0)$, Alice measures her half of the state in \eqref{eq: starting state} in the computational basis.
    \item Upon receiving a question with prefix ``$\sim$tCHSH($\frac{1}{\sqrt{2}}$)'', Alice and Bob first apply embezzling unitaries $W_{\tilde{A}A'}$ and $W_{\tilde{B}B'}$ respectively, such that (approximately)
  $\frac{1}{\sqrt{2}}(\ket{11}+ \ket{22}) \mapsto \ket{11}$ and $\ket{00} \mapsto \ket{00}$. So the resulting state is
  \begin{equation}
  \label{eq: state 2}
      \sqrt{\frac23}\left(\frac{1}{\sqrt{2}}\ket{00}+\ket{11}\right)_{\tilde{A}\tilde{B}} \otimes \ket{\Gamma_d}_{A'B'}.
  \end{equation}
 They then perform the corresponding ideal measurements for $\textnormal{$\sim$tCHSH}(\frac{1}{\sqrt{2}})$ on registers $\tilde{A}, \tilde{B}$ (where Alice takes the role of the second player, and Bob takes the role of the first player). In particular, on question $(\textnormal{``$\sim$tCHSH($\frac{1}{\sqrt{2}}$)''}, 0)$, Bob measures his half of the state in \eqref{eq: state 2} in the computational basis.
\end{itemize}
A key observation is that when Alice and Bob are asked questions $(\textnormal{``3-CHSH''}, 0)$ and $(\sim\textnormal{``tCHSH($\frac{1}{\sqrt{2}}$)''}, 0)$ respectively, then it is straightforward to see that, if they follow the above strategy, they reply with answers $(a,b)$ which attain a score of $1$ in part (c) of Fig. \ref{fig: our non-local game}, i.e. $(a,b) \in \{(0,0)\} \cup \left(\{1,2\} \times \{1,2\}\right)$. 

Next, we define the players' ideal measurements precisely. Recall the isometry $V: \C^2 \to \C^3$ defined in subsection \ref{sec: non-local games examples} as follows:
\begin{equation}
	V \ket 0 = \ket {1},\,\,\,
    V \ket 1 = \ket {2}
\end{equation}
Recall also the notation introduced in subsection \ref{sec: non-local games examples} along with $V$. In particular, we write $V(O)$ to refer to the pushforward $VO V^\dagger$ of $O$ along $V$. For $O$ an operator with $+1,0,-1$ eigenvalues, we write $O^+$ for the projection onto the $+1$ eigenspace and $O^-$ for the projection onto the $-1$ eigenspace. If $\m H_A \approx \C^3$, we still write $\sigma^z_A$ to mean $\sigma^z_A = \proj 0_A - \proj 1_A$. On the other hand, in accordance with the notation above, we write $V(\sigma^z)_A$ to mean  $V(\sigma^z)_A = \proj 1_A - \proj 2_A$. 

 Let $\{\ket{\Gamma_d}_{ABA'B'}\}$ be the embezzling family from Example \ref{ex: emb family}, and $W_{AA'}: (\C^2)_{A} \otimes (\C^2)^{\otimes d}_{A'} \rightarrow (\C^2)_{A} \otimes (\C^2)^{\otimes d}_{A'}$, $W_{BB'}:(\C^2)_{B} \otimes (\C^2)^{\otimes d}_{B'} \rightarrow (\C^2)_{A} \otimes (\C^2)^{\otimes d}_{A'}$ be the left-shift unitaries over the $d+1$ qubit registers. Define $\tilde{W}_{\tilde{A}A'}: (\C^3)_{\tilde{A}} \otimes (\C^2)^{\otimes d}_{A'} \rightarrow (\C^3)_{\tilde{A}} \otimes (\C^2)^{\otimes d}_{A'}$ as $$\tilde{W}_{\tilde{A}A'} = \left(\proj 0_{\tilde{A}} \otimes I_{A'} \right) \oplus [(V \otimes I) W_{AA'}(V^{\dagger} \otimes I)],$$ and define $\tilde{W}_{\tilde{B}B'}$ analogously.

The following is the family of ideal strategies for $G_{emb}$ achieving a value arbitrarily close to $\frac13 (\omega^*_{3\textnormal{-CHSH}}+\omega^*_{\sim\textnormal{tCHSH}(\frac{1}{\sqrt{2}})} + 1)$.

\begin{definition}[Ideal strategy for $G_{emb}$]
\label{def: ideal strategy}
The family of ideal strategies is $\{S_d\}_{d \in \mathbb{N}}$, with $S_d = (\ket{\Psi_d}, \{P_x^a\}, \{Q_y^b\})$, where $$\ket{\Psi_d} = \frac{1}{\sqrt{3}}(\ket{00}+\ket{11}+\ket{22})_{\tilde{A}\tilde{B}} \otimes \ket{\Gamma_d}_{A'B'},$$ and the ideal measurements are described in Tables \ref{tab:alice-ideal} and \ref{tab:bob-ideal}.

\begin{table}[H]
    \caption{Alice's ideal measurements for $G_{emb}$. The entry in cell $x,a$ is the projector $P_{x}^a$ (tensored identities are implied where omitted, and $P_{rest}$ completes the set of orthogonal projections in a row).}
    \label{tab:alice-ideal}
    \centering
    \begin{tabular}{| c || Sc | Sc | Sc | Sc | Sc |}
	\hline
	\diagbox[width=2.5em]{$x$}{$a$} & 0 & 1 & 2 \\ \hhline{|====|}
	(``3-CHSH'', $0$)  & $\proj 0_{\tilde{A}}$ & $\proj 1_{\tilde{A}}$ & $\proj 2_{\tilde{A}}$ \\ \hline
	(``3-CHSH'', $1$) & $(\sigma^x)^+_{\tilde{A}}$ & $(\sigma^x)^-_{\tilde{A}}$ & $\proj 2_{\tilde{A}}$ \\ \hline
	(``3-CHSH'', $2$) & $\proj 0_{\tilde{A}}$ & $[V(\sigma^x)]^+_{\tilde{A}}$ & $[V(\sigma^x)]^-_{\tilde{A}}$ \\ \hline
	(``$\sim$tCHSH($\frac{1}{\sqrt{2}}$), $0$) & $W_{\tilde{A}A'}^{\dagger}([\sigma^x( \sigma^z_{\alpha = \frac{1}{\sqrt{2}}})^-\sigma^x]_{\tilde{A}} \otimes I_{A'}) W_{\tilde{A}A'} $ & $W_{\tilde{A}A'}^{\dagger}([\sigma^x( \sigma^z_{\alpha = \frac{1}{\sqrt{2}}})^+\sigma^x]_{\tilde{A}} \otimes I_{A'}) W_{\tilde{A}A'} $ & $P_{rest}$\\
	\hline
	(``$\sim$tCHSH($\frac{1}{\sqrt{2}}$), $1$) & $W_{\tilde{A}A'}^{\dagger}([\sigma^x( \sigma^x_{\alpha = \frac{1}{\sqrt{2}}})^-\sigma^x]_{\tilde{A}} \otimes I_{A'}) W_{\tilde{A}A'} $ & $W_{\tilde{A}A'}^{\dagger}([\sigma^x( \sigma^x_{\alpha = \frac{1}{\sqrt{2}}})^+\sigma^x]_{\tilde{A}} \otimes I_{A'}) W_{\tilde{A}A'} $ & $P_{rest}$ \\
	\hline
\end{tabular}
\end{table}

\begin{table}[H]
    \centering
    \caption{Bob's ideal measurements for $G_{emb}$. The entry in cell $y,b$ is the projector $P_{y}^b$ (tensored identities are implied where omitted, and $P_{rest}$ completes the set of orthogonal projections in a row).}
    \label{tab:bob-ideal}
    \begin{tabular}{| c || Sc | Sc | Sc | Sc | Sc |}
	\hline
	\diagbox[width=2.5em]{$y$}{$b$} & 0 & 1 & 2 \\ \hhline{|====|}
	(``3-CHSH'', $0$)  & $(\sigma^z_{\alpha = 1})^+_{\tilde{B}}$ & $(\sigma^z_{\alpha = 1})^-_{\tilde{B}}$ & $\proj 2_{\tilde{B}}$ \\ \hline
	(``3-CHSH'', $1$)  & $(\sigma^x_{\alpha = 1})^+_{\tilde{B}}$ & $(\sigma^x_{\alpha = 1})^-_{\tilde{B}}$ & $\proj 2_{\tilde{B}}$ \\ \hline
	(``3-CHSH'', $2$)  & $\proj 0_{\tilde{B}}$ & $[V(\sigma^z_{\alpha=1})]^+_{\tilde{B}}$ & $[V(\sigma^z_{\alpha=1})]^-_{\tilde{B}}$ \\ \hline
	(``3-CHSH'', $3$)  & $\proj 0_{\tilde{B}}$ & $[V(\sigma^x_{\alpha=1})]^+_{\tilde{B}}$ & $[V(\sigma^x_{\alpha=1})]^-_{\tilde{B}}$ \\ 
	\hline
	(``$\sim$tCHSH($\frac{1}{\sqrt{2}}$), $0$) & $W_{\tilde{B}B'}^{\dagger}([\sigma^x( \sigma^z)^-\sigma^x]_{\tilde{B}} \otimes I_{B'}) W_{\tilde{B}B'} $ & $W_{\tilde{B}B'}^{\dagger}([\sigma^x( \sigma^z)^+\sigma^x]_{\tilde{B}} \otimes I_{B'}) W_{\tilde{B}B'} $ & $P_{rest}$ \\ \hline
	(``$\sim$tCHSH($\frac{1}{\sqrt{2}}$), $1$) & $W_{\tilde{B}B'}^{\dagger}([\sigma^x( \sigma^x)^-\sigma^x]_{\tilde{B}} \otimes I_{B'}) W_{\tilde{B}B'} $ & $W_{\tilde{B}B'}^{\dagger}([\sigma^x( \sigma^x)^+\sigma^x]_{\tilde{B}} \otimes I_{B'}) W_{\tilde{B}B'} $ & $P_{rest}$ \\ \hline
\end{tabular}
\end{table}

\end{definition}

Note that the ideal operators for the equivalent version of tilted CHSH from Theorem \ref{thm: tilted chsh self testing sim} are the same as for the original tilted CHSH of Definition \ref{def: tchsh as game}, except conjugated by $\sigma_x$ (i.e. the answers are flipped).

\begin{prop}[Completeness]
\label{prop: completeness}
Let $\{S_d\}_{d\in \mathbb{N}}$ be the family of strategies from Definition \ref{def: ideal strategy}, and $G_{emb}$ the non-local game from Fig. \ref{fig: our non-local game}. Then, $
\omega(S_d, G_{emb}) = \omega^*(G_{emb})- O(\frac{1}{d})$.
\end{prop}
\begin{proof}
The value of strategy $S_d$ in part (a) is exactly $\omega^*_{3\textnormal{-CHSH}}$. This is because the starting state is the ideal state for $\omega^*_{3\textnormal{-CHSH}}$ and measurements are the ideal ones from Definition \ref{def: ideal 3-chsh}. The value in part (b) is $\omega^*_{\sim\textnormal{tCHSH}(\frac{1}{\sqrt{2}})} - O(\frac{1}{d})$. This is because the joint state resulting from the embezzling transformation has fidelity $1-O(\frac{1}{d})$ with the ideal state for $\textnormal{$\sim$tCHSH}(\frac{1}{\sqrt{2}})$ (from Theorem \ref{thm: tilted chsh self testing sim}), and the measurements for part (b) are also ideal. The value in part (c) is easily seen to be exactly $1$. Thus, $\omega(S_d, G_{emb}) = \frac13 (\omega^*_{3\textnormal{-CHSH}}+\omega^*_{\sim\textnormal{tCHSH}(\frac{1}{\sqrt{2}})} + 1) - O(\frac{1}{d})$. Together with the upper bound in the proof of Proposition \ref{prop: value}, this completes the proof of Proposition \ref{prop: value} (i.e. $\omega^*(G_{emb}) = \frac13 (\omega^*_{3\textnormal{-CHSH}}+\omega^*_{\sim\textnormal{tCHSH}(\frac{1}{\sqrt{2}})} + 1)$), and gives $\omega(S_d, G_{emb}) = \frac13 (\omega^*_{3\textnormal{-CHSH}}+\omega^*_{\sim\textnormal{tCHSH}(\frac{1}{\sqrt{2}})} + 1) - O(\frac{1}{d})$, as desired.
\end{proof}

\section{Soundness}
\label{sec: soundness}
\begin{theorem}
\label{thm: main}
There exists a constant $C>0$ such that any quantum strategy $S$ for the game $G_{emb}$ of Fig. \ref{fig: our non-local game} with value $\omega(S, G_{emb}) \geq \omega^*(G_{emb}) - \epsilon$, for some $0<\epsilon < C$, must have dimension $2^{\Omega(\epsilon^{-1/8})}$.
\end{theorem}
The proof of Theorem \ref{thm: main} can be broken down into two parts:
\begin{itemize}
    \item[(i)] First, we will show that performing well in parts $\textbf{(a)}, \textbf{(b)}$ and $\textbf{(c)}$ of the game imposes a certain structure on the strategy of the provers. 
    \item[(ii)] Second, we show that such a structured strategy can be used to play well also in the ``coherent state exchange'' game of Leung, Toner and Watrous \cite{leung2013coherent}. This reduction allows us to translate the lower bounds on the dimension of an approximately optimal strategy in the ``'coherent state exchange'' game to lower bounds on the dimension of an approximately optimal strategy for our game.
\end{itemize}




\begin{proof}[Proof of Theorem \ref{thm: main}]
Let $\left(\ket{\psi} \in \mathcal{H}_A \otimes \mathcal{H}_B, \{P_x^a\}, \{Q_y^b\}\right)$ be a strategy whose value in $G_{emb}$ is $\epsilon$-close to $\omega^*(G_{emb}) = \frac13 ( w_{\textnormal{3CHSH}}^* + w_{\textnormal{2CHSH}}^* + 1 )$. This implies that, for each part of the game, the strategy's expected score is $3\epsilon$-close to optimal. From each part we deduce the following: 
\begin{itemize}
    \item[\textbf{(a)}] From Theorem \ref{thm: self-testing gen chsh} (the case $d=3$), upon picking an appropriate constant $C>0$, there exists a local unitary $U: \mathcal{H}_A \otimes \mathcal{H}_B \rightarrow (\mathbb{C}^3)_{A_1} \otimes (\mathbb{C}^3)_{B_1} \otimes \mathcal{H}_{A'} \otimes \mathcal{H}_{B'}$, and an auxiliary state $\ket{aux} \in \mathcal{H}_{A'} \otimes \mathcal{H}_{B'}$ such that \begin{itemize}
        \item $\ket{\psi} \approx_{U, O(\epsilon^{1/8})} \frac{1}{\sqrt{3}}(\ket{00} + \ket{11} + \ket{22}) \otimes \ket{aux}$
        \item $P^0_{(\textnormal{``$3$-CHSH''}, 0)} \ket{\psi} \approx_{U,O(\epsilon^{1/8})} \frac{1}{\sqrt{3}}\ket{00} \otimes \ket{aux}$ 
    \end{itemize}
    \item[\textbf{(b)}] From Theorem \ref{thm: tilted chsh self testing sim}, there exists a local unitary $U': \mathcal{H}_A \otimes \mathcal{H}_B \rightarrow (\mathbb{C}^2)_{A_2} \otimes (\mathbb{C}^2)_{B_2} \otimes \mathcal{H}_{A''} \otimes \mathcal{H}_{B''}$, and an auxiliary state $\ket{aux'} \in \mathcal{H}_{A''} \otimes \mathcal{H}_{B''}$ such that \begin{itemize}
        \item $\ket{\psi} \approx_{U', O(\epsilon^{1/2})} \frac{1}{\sqrt{3}}(\ket{00} + \sqrt{2} \ket{11}) \otimes \ket{aux'}$
        \item $Q^0_{(\textnormal{``$\sim$tCHSH($\frac{1}{\sqrt{2}}$)''}, 0)} \ket{\psi} \approx_{U',O(\epsilon^{1/2})} \frac{1}{\sqrt{3}}\ket{00} \otimes \ket{aux'}$ 
    \end{itemize}
    \item[\textbf{(c)}] $P^0_{(\textnormal{``$3$-CHSH''}, 0)} \ket{\psi} \approx_{O(\epsilon^{1/2})} Q^0_{(\textnormal{``$\sim$tCHSH($\frac{1}{\sqrt{2}}$)''}, 0)} \ket{\psi}$
\end{itemize}
Notice that $\textbf{(a)}, \textbf{(b)}, \textbf{(c)} \implies$ the local unitary $\tilde{U} := (U')(U)^{-1} : (\mathbb{C}^3)_{A_1} \otimes (\mathbb{C}^3)_{B_1} \otimes \mathcal{H}_{A'} \otimes \mathcal{H}_{B'} \rightarrow (\mathbb{C}^2)_{A_2} \otimes (\mathbb{C}^2)_{B_2} \otimes \mathcal{H}_{A''} \otimes \mathcal{H}_{B''}$ is such that 
\begin{equation}
\label{eq: 1}
    \frac{1}{\sqrt{3}}(\ket{00} + \ket{11} + \ket{22}) \otimes \ket{aux} \approx_{\tilde{U}, O(\epsilon^{1/8})} \frac{1}{\sqrt{3}}(\ket{00} + \sqrt{2}\ket{11}) \otimes \ket{aux'},
\end{equation}
and moreover
\begin{equation}
\label{eq: 2}
\frac{1}{\sqrt{3}}\ket{00} \otimes \ket{aux} \approx_{\tilde{U}, O(\epsilon^{1/8})} \frac{1}{\sqrt{3}}\ket{00} \otimes \ket{aux'}.
\end{equation}

\eqref{eq: 1} and \eqref{eq: 2} immediately imply that \begin{equation}
\label{eq: main}
    \frac{1}{\sqrt{2}} (\ket{11} + \ket{22}) \otimes \ket{aux} \approx_{\tilde{U}, O(\epsilon^{1/8})} \ket{11} \otimes \ket{aux'}.
\end{equation}
We claim that the local unitary $\tilde{U}$ can be used to approximately win the ``coherent state exchange'' game of Leung, Toner and Watrous \cite{leung2013coherent}. More precisely, since Equation \eqref{eq: main} is $O(\epsilon^{1/8})$-approximate (with respect to Euclidean norm), we claim that one can construct a strategy which employs $\tilde{U}$, and in which the provers' initial state is $\ket{aux}$, which wins the game of \cite{leung2013coherent} with probability $1-O(\epsilon^{1/4})$. Assuming this claim is true, the rest of the proof is straightforward: it was shown in \cite{leung2013coherent} that the winning probability of any strategy in the ``coherent state exchange game'' is upper bounded by $1-\frac{1}{32\log^2(3d)}$, where $d$ is the dimension of the states used; this implies that it must be $$\frac{1}{32\log^2(3d)} = O(\epsilon^{1/4}) \implies d = 2^{\Omega(\epsilon^{-\frac{1}{8}})}.$$
To conclude the proof of Theorem \ref{thm: main}, we prove the above claim. 

The ``coherent state exchange'' game of \cite{leung2013coherent} between a quantum referee and two non-communicating provers, proceeds as follows: \begin{itemize}
    \item The referee initializes a qubit register $\textsf{R}$ and qutrit registers $\textsf{S}$ and $\textsf{T}$ in the state
\begin{equation}
\label{eq: state}
\frac{1}{\sqrt{2}}(\ket{0}_{\textsf{R}}\ket{00}_{\textsf{S}\textsf{T}} + \ket{1}_{\textsf{R}}\ket{\phi^+}_{\textsf{S}\textsf{T}}),
\end{equation}
where $\ket{\phi^+} = \frac{1}{\sqrt{2}} (\ket{00}+ \ket{11})$. The referee sends registers $\textsf{S}$ and $\textsf{T}$ to Alice and Bob respectively. 
\item The referee receives single-qubit registers $\textsf{A}$ and $\textsf{B}$ from Alice and Bob respectively. The triple $(\textsf{R}, \textsf{A}, \textsf{B})$ is measured with projective measurement $\{\Pi_0, \Pi_1\}$, where $\Pi_0 = I - \ket{\gamma}\bra{\gamma}$ and $\Pi_1 = \ket{\gamma}\bra{\gamma}$, and $\ket{\gamma} = \frac{1}{\sqrt{2}} (\ket{000}+ \ket{111})$.
\end{itemize}
Consider the following strategy of the provers for this game. They start by sharing the state $\ket{aux} \in \mathcal{H}_{A'} \otimes \mathcal{H}_{B'}$. Upon receiving the qutrit registers $\textsf{S}$ and $\textsf{T}$ of the state \eqref{eq: state}, they apply $\tilde{U}$ to registers $(\mathbb{C}^3)_{\textsf{S}} \otimes (\mathbb{C}^3)_{\textsf{T}} \otimes \mathcal{H}_{A'} \otimes \mathcal{H}_{B'}$ (up to relabelling registers $A_1$ and $B_1$ as $\textsf{S}$ and $\textsf{T}$), obtaining a state in $(\mathbb{C}^2)_{A_2} \otimes (\mathbb{C}^2)_{B_2} \otimes \mathcal{H}_{A''} \otimes \mathcal{H}_{B''}$. Equations \eqref{eq: 2} and \eqref{eq: main} imply that the resulting state on registers $\textsf{R}, A_2, B_2, A'', B''$ is $O(\epsilon^{1/8})$-close to $\frac{1}{\sqrt{2}}(\ket{000} + \ket{111}) \otimes \ket{aux'}$. And hence the state on $\textsf{R}, A_2, B_2$ is $O(\epsilon^{1/8})$-close to the desired state (in Euclidean norm). Qubit registers $A_2$ and $B_2$ are then sent back to the referee as $\textsf{A}$ and $\textsf{B}$. Converting the $O(\epsilon^{1/8})$-closeness to a probability of winning in the game, gives a lower bound of $1-O(\epsilon^{1/4})$, and thus concludes the proof. We remark that the self-testing arguments applied in this proof also hold when the uncharacterized strategy of the provers is not a priori assumed to be finite-dimensional, but is instead allowed to infinite-dimensional (on separable Hilbert spaces).

\end{proof}

\section{Non-closure of the set of quantum correlations}
\label{sec: non-closure}

A corollary of Proposition \ref{prop: completeness} and Theorem \ref{thm: main} (completeness and soundness for our game) is that the set $\m C_{qs}$ of quantum correlations induced by, possibly infinite-dimensional, quantum strategies (in the tensor product model) is not closed, i.e. $\m C_{qs} \neq \m C_{qa}$, where the latter is the closure. For precise definitions of these sets see \cite{coladangelo2018unconditional}. We use superscripts to denote question and answer set sizes. For instance $\m C_{qs}^{m,n,r,s}$ is on question sets of size $m,n$ and answer sets of size $r,s$.

\begin{cor}
$\m C_{qs}^{5,6,3,3} \neq  \m C_{qa}^{5,6,3,3}$.
\end{cor}
\begin{proof}
In the proof of Theorem \ref{thm: main}, we argued that any strategy with value $\omega^*(G_{emb}) - \epsilon$ in our game $G_{emb}$ can be used to construct a strategy that embezzles an EPR pair into a product state, up to $O(\epsilon^{1/8})$ error in Euclidean norm. This implies that no strategy has value exactly $\omega^*(G_{emb})$. Suppose otherwise for a contradiction. Then, by the reduction in the proof of Theorem \ref{thm: main}, we can construct a strategy that wins the game of \cite{leung2013coherent} with probability $1$. From \cite{leung2013coherent}, this is known to imply existence of a strategy that embezzles perfectly (the argument that shows this implication in \cite{leung2013coherent} is phrased for finite-dimensional strategies, but it holds also for infinite-dimensional ones). A perfect embezzling strategy consists of a state $\ket{\Psi} \in \m H_{A'} \otimes \m H_{B'}$ and a local unitary $U = U_{AA'} \otimes U_{BB'}$ such that $U \ket{\phi^+}_{AB} \otimes \ket{\Psi}_{A'B'} = \ket{00}_{AB} \otimes \ket{\Psi}_{A'B'}$. Since Schmidt coefficients are preserved under local unitaries, it is clear that, whatever the Schmidt coefficients of $\ket{\Psi}$ are, the Schmidt coefficients of the LHS and RHS are different. This gives a contradiction.

On the other hand, Proposition \ref{prop: completeness} gives a sequence of strategies whose value tends to $\omega^*(G_{emb})$. If one considers the sequence of correlations induced by such strategies, it is clear that such a sequence has a limit, and that the limiting correlation has value $\omega^*(G_{emb})$. Such a limiting correlation is thus in $\m C_{qa}^{5,6,3,3}$ but not in $\m C_{qs}^{5,6,3,3}$.

\end{proof}

We emphasize that strictly stronger separations (for question sets of size $5$ and answer sets of size $2$) are known \cite{dykema2017non, musat2018non}. The latter appeared after the original breakthrough proof of Slofstra, for much larger question and answer sets \cite{slofstra2019set}. What stands out about our proof is that, unlike all previous proofs, it does not involve any representation theoretic machinery. 

\section*{Acknowledgements}
The author thanks William Slofstra and Thomas Vidick for useful comments on an earlier version of this paper. The author also thanks the latter for helpful discussions. The author thanks Vern Paulsen for a useful email exchange about the game in \cite{dykema2017non}.
The author was supported by the Kortschak Scholars program and AFOSR YIP award number FA9550-16-1-0495.

\bibliographystyle{halpha}
\bibliography{references}

\end{document}